\documentclass[letterpaper, 10 pt, conference]{ieeeconf}  

\IEEEoverridecommandlockouts                              
\usepackage{bm}
\usepackage[utf8]{inputenc}
\usepackage[T1]{fontenc}
\usepackage{graphicx}
\graphicspath{ {./images/} }

\usepackage{sgame}
\usepackage{amssymb}
\usepackage{amsmath}
\usepackage{xcolor}
\usepackage{float}
\usepackage{comment}
\usepackage{lipsum}
\usepackage{mathbbol}
\usepackage[capitalize]{cleveref}
\usepackage{enumerate}
\usepackage{cite}
\usepackage{mathtools}

\def \T {\mathsf{T}}

\def \E {\mathbf{E}}

\def \< {\langle}
\def \> {\rangle}

\def \vw {\bm{w}}

\newtheorem{remark}{Remark}
\newtheorem{proposition}{Proposition}
\newtheorem{theorem}{Theorem}

\newtheorem{corollary}{Corollary}

\newtheorem{assumption}{Assumption}

\DeclareMathOperator{\Equaldef}{\overset{def}{=}}

\title{\LARGE \bf
Neuroscheduling for Remote Estimation}

\author{Marcos M. Vasconcelos and Yifei Zhang  
\thanks{M. M. Vasconcelos and Y. Zhang are with the Department of Electrical and Computer Engineering, FAMU-FSU College of Engineering, Florida State University,  Tallahassee, FL 32306, USA. E-mails:
        {\tt   m.vasconcelos@fsu.edu, yz23r@fsu.edu@fsu.edu}.}%
}

\begin{document}

\maketitle
\thispagestyle{empty}
\pagestyle{empty}

\begin{abstract}%
Many modern distributed systems consist of devices that generate more data than what can be transmitted via a communication link in near real time with high-fidelity. We consider the \textit{scheduling problem} in which a device has access to multiple data sources, but at any moment, only one of them is revealed in real-time to a remote receiver. Even when the sources are Gaussian, and the fidelity criterion is the mean squared error, the globally optimal data selection strategy is not known. We propose a data-driven methodology to search for the elusive optimal solution using linear function approximation approach called \textit{neuroscheduling} and establish necessary and sufficient conditions for the optimal scheduler to not over fit training data. 
Additionally, we present several numerical results that show that the globally optimal scheduler and estimator pair to the Gaussian case are nonlinear.
\end{abstract}

\vspace{10pt}


\section{Introduction}

A fundamental question in networked systems is: \textit{if an agent has more data than what can be transmitted over a link at any given time, how does it choose what to transmit and what to discard?} The transmitter's choice depends on what the receiver wants and what it is capable of doing with the data transmitted. We approach this problem by first formulating it in its simplest form, shown in \cref{fig:scheduling}. One transmitting agent has access to two variables $X_1$ and $X_2$, and must choose one of them to send to a receiving agent according to a function, $\gamma$. The receiving agent will only have access to either $X_1$ or $X_2$, but it wants to recover both. To that effect, the receiver uses a function, $\eta$. The goal of the designer is to optimize an appropriate performance metric (mean squared error, probability of error, etc.) over all possible functions $\gamma$ and $\eta$.

To the best of our knowledge, the optimal solution to this problem is unknown even in the canonical case of Gaussian sources under the mean-squared error distortion. In this paper, we propose searching for the optimal scheduler using data-driven linear function approximation approach called \textit{neuroscheduling}. Our main contributions are: 1. we provide a condition that needs to be satisfied by the basis functions in our approximation to avoid over-fitting; 2. we show via numerical results that linear estimators in the Gaussian case are not optimal.

\subsection{Related Literature}
Data selection is a topic with a rich history rooted in statistics, where the goal is choose a subset of the original feature set without making any transformation on the attributes \cite{Villa:2023}. The conventional approach to this is to use informativeness metrics reminiscent of Estimation theory, such as functions of the error covariance matrix and select a subset of  predetermined size $K$ from a larger set with $N$ features \cite{Hashemi:2020}. This optimization problem is known to be computationally intractable and the alternative is to use suboptimal algorithms that have an acceptable performance and can produce a solution in reasonable time \cite{Joshi:2008,Moon:2017,Takahashi:2023}. However, this approach is done \textit{before the data is collected} based on the probabilistic model of the data and the observation noise. Moreover, the strategies obtained using such approaches are static, e.g. ``\textit{if sensor A is more informative than sensor B, then sensor A is prioritized}.''

Information theory provides alternative ways to measure informativeness, that can also be used to perform data selection \cite{Wang:2022}. However, information theoretic quantities, such as entropy and mutual information, do not depend on the data content at all and were designed to capture features related to the probabilistic model. Recent developments have established a connection between information and moments of the underlying random variables \cite{Alghamdi:2024}, but their use to select informative data on the basis of their realizations remains an open question. Existing approaches require a full knowledge of the probabilistic model of the system and provide policies that do not capture situations where a less informative data source may produce a very informative data sample.

\begin{figure}[t!]
    \centering
\includegraphics[width=0.4\textwidth]{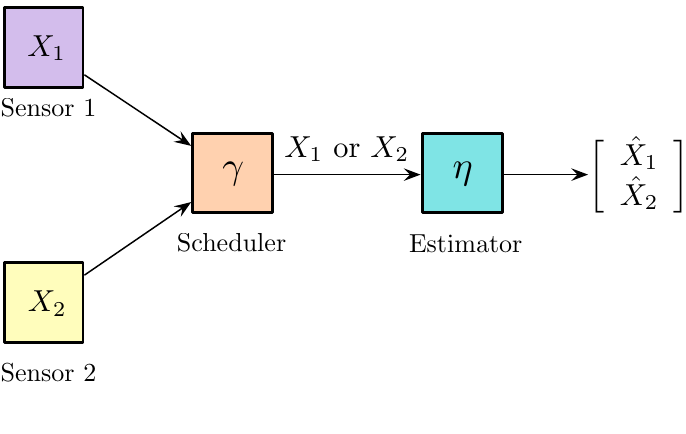}
    \caption{Observation-driven sensor scheduling system with two sensors, one scheduler and one receiver.}
    \label{fig:scheduling}
\end{figure}

Control theory, on the other hand, is based on the notion that data is measured first and then processed using policies -- therefore, the results are naturally adapted to data. In the control theory literature, data selection appear under the moniker of \textit{event-triggered} estimation/control/learning \cite{Molin:2017,Tabuada:2007,Solowjow:2020,Yun:2023,Salimnejad:2024}. The premise is that in an event-triggered policy, the data is transmitted to a remote decision-maker if the observed variable exceeds a certain threshold. The idea is that observations that are small (in an appropriate sense) are predictable and redundant -- they are not informative. A large body of literature based on this principle exists. In particular, the area of remote estimation has emerged as an area of interest, where smart sensors independently make decisions to transmit (or not) their observations to a fusion center \cite{Lipsa:2011,Nayyar:2013,Vasconcelos:2017a,Vasconcelos:2019b,Ni:2021,Leong:2023}. The literature reveals that in existing applications, each sensor observes a single variable/vector that is either fully kept or discarded. However, to the best of our knowledge, a similar systematic event-based approach to subset selection where multiple variables are simultaneously observed at the same location and a subset of them are selected has not been developed until now.

Another body of literature closely related to our work pertains to the Value of Information (VoI), a concept originally introduced in the field of economics. Historically, VoI has been considered in making decisions in optimal control settings \cite{Chiariotti:2022,Soleymani:2022,Soleymani:2023,Zancanaro:2023}, and it is similar to what we propose here: transmit the data (or not) based on the level of performance improvement that can be achieved using that additional data. In certain cases, scheduling based on VoI may coincide with the scheduling of Extremum Information \cite{Vasconcelos:2020a}. It has been well-documented that the decision not to communicate an observation may convey implicit information \cite{Imer:2010,Lipsa:2011,Vasconcelos:2017a,Zhang:2023a,Santi:2024}. Such implicit communication leads to substantial gains in performance, which should also be taken into consideration when computing the VoI. 
However, the framework introduced in \cite{Soleymani:2022,Soleymani:2023} deals only with a single data source, whereas our focus in this paper  is on identifying a one of multiple sources that have the largest VoI for a particular task.

\section{Problem setup}

Consider two Gaussian random variables observed by a scheduler, which is interested in sending one of them to a receiver jointly distributed according to $(X_1,X_2) \sim f_{X_1X_2}$, and assume that all sources have stationary statistics. At a given time, the scheduler observes one realization of $X_1$ and $X_2$ decides using a scheduling policy, whether $X_1$ or $X_2$ will be transmitted to the destination. Communication happens in real time, i.e., we do not make our decisions based on observing blocks of data, which would necessarily incur communication delay \cite{Akyol:2014}. A \textit{scheduling policy} is a map from the observation space $\mathbb{R}^2$ to the set $\{1,2\}.$ The decision variable $U\in\{1,2\}$ is computed according to

\begin{equation}
U = \gamma(X_1,X_2).
\end{equation}

Based on $U$, the information sent over the channel is determined as follows:

\begin{equation}
Y = \begin{cases}
(1,X_1), & \textup{if} \ \ U =1\\
(2,X_2), & \textup{if} \ \ U =2,
\end{cases}
\end{equation}

where the index in front of the information variable is important since the receiver needs to know which information source generated the real number observed in the packet, \textit{i.e.}, it indicates the origin of the communication packet.

At the destination, the receiver implements an \textit{estimation policy}, which attempts to reconstruct both sources based on the observation received over the link, $Y$. We define $\eta$ as such estimation policy, and the estimates are computed as follows:
\begin{equation}
    \begin{bmatrix}
\hat{X}_1 \\
\hat{X}_2
    \end{bmatrix} = \eta(Y).
\end{equation}

From the designer's perspective, the optimization problem that we are interested in solving is the following:
\begin{equation}\label{OC}
\min_{(\gamma,\eta)\in \Gamma\times \mathrm{H}} \mathbf{E} \left[ (X_1-\hat{X}_1)^2 + (X_2-\hat{X}_2)^2  \right],
\end{equation}
where $\Gamma$ and $H$ are the spaces of all \textit{admissible} scheduling and estimation policies, respectively.
By looking at Problem \eqref{OC}, we are searching for an optimal scheduling-estimation pair. A conceptualy simpler version of this problem is: given an estimator pair $(\eta_1,\eta_2)$, determine the scheduling policy that is optimal for this fixed estimator. The answer is the following scheduling policy:
\begin{equation}\label{max}
\gamma^\star_{\eta_1,\eta_2}(x_1,x_2) = \begin{cases}
1, & \text{if} \ \ |x_1-\eta_1(x_2)| \geq |x_2-\eta_2(x_1)|\\
2, & \text{otherwise,}
\end{cases}
\end{equation}
where $\eta_1$ and $\eta_2$ are functions used to estimate $X_1$ and $X_2$, respectively, based on the information revealed to the receiver by the scheduler. Assuming that this scheduling policy is used, then the next step is to find pair of optimal estimators $(\eta_1^*,\eta_2^*)$,
which are  nonlinear functions, in general. The scheduling policy in \cref{max} is called \textit{max-scheduling} \cite{Vasconcelos:2020}.  In many cases, it is possible to show that certain estimation policies are person-by-person optimal \cite{Yuksel:2013} with this policy, leading to locally optimal solutions. However, such results rely on a \textit{guess-and-verify} approach, and there is no systematic way of obtaining locally nor globally optimal solutions.


\subsection{Optimality of Estimator Pairs}

Generally, a scheduler is a measurable function on $\mathbb{R}^2$ that induces a partition of $\mathbb{R}^2$. Let $\gamma : \mathbb{R}^2 \rightarrow \{1,2\}$, where $\{(x_1,x_2)\in\mathbb{R}^2 \mid \gamma(x_1,x_2)=1\} \cup \{(x_1,x_2)\in\mathbb{R}^2 \mid \gamma(x_1,x_2)=2\} = \mathbb{R}^2$ and $\{(x_1,x_2)\in\mathbb{R}^2 \mid \gamma(x_1,x_2)=1\} \cap \{(x_1,x_2)\in\mathbb{R}^2 \mid \gamma(x_1,x_2)=2\} = \varnothing$. We are interested in obtaining the optimal estimator pair when the scheduler follows the max-scheduling policy of \cref{max}. For this purpose, let us define the error vector on $\mathbb{R}^2$ for a particular estimator pair $(\eta_1,\eta_2)$: 
\begin{equation}\label{EV}
    \epsilon_{\eta_1,\eta_2} \Equaldef \begin{cases}
        X_2-\eta_2(X_1), & U = 1 \\
        X_1-\eta_1(X_2), & U = 2.
    \end{cases}
\end{equation}

For a given pair $\eta_1,\eta_2$, the quantity $\epsilon_{\eta_1,\eta_2}$ is a random variable on $\mathbb{R}^2$. Whenever we change the estimator function pair $(\eta_1,\eta_2)$, its distribution changes. Ideally, the following estimator pair should be the optimal one we are after:
\begin{equation}
\eta_1^\star(x_2) = \mathbf{E} \big[X_1 \mid X_2 = x_2, U=2 \big],
\end{equation}
and
\begin{equation}
\eta_2^\star(x_1) = \mathbf{E} \big[X_2 \mid X_1 = x_1, U=1 \big].
\end{equation}

It is important to notice that, in general, $\mathbf{E} \big[X_1 \mid X_2 = x_2, U=2 \big] \neq \mathbf{E} \big[X_1 \mid X_2 = x_2 \big]$ and $\mathbf{E} \big[X_2 \mid X_1 = x_1, U=1 \big] \neq \mathbf{E} \big[X_2 \mid X_1 = x_1 \big]$. Take the Gaussian case as an example, suppose $(X_1,X_2)$ follows a standard bivariate normal distribution with nonzero correlation coefficient $\rho$, then define
\begin{equation}
\hat\eta_1(x_2) = \mathbf{E} \big[X_1 \mid X_2 = x_2 \big] = \rho x_2 
\end{equation}
and
\begin{equation}
\hat\eta_2(x_1) = \mathbf{E} \big[X_2 \mid X_1 = x_1 \big] = \rho x_1.
\end{equation}
Computing the expectation of $\epsilon_{\hat\eta_1,\hat\eta_2}$,we get
\begin{multline}
    \mathbf{E}[\epsilon_{\hat\eta_1,\hat\eta_2}] = \frac{1}{2}\int_{\{(x_1,x_2)\in\mathbb{R}^2 \mid \gamma(x_1,x_2)=1\}} -\rho x_2 \hspace{5pt} \mathrm{d}\mu_X \\ + \frac{1}{2}\int_{\{(x_1,x_2)\in\mathbb{R}^2 \mid \gamma(x_1,x_2)=2\}} -\rho x_1  \hspace{5pt} \mathrm{d}\mu_X,
\end{multline}
where $\mu_X$ is the measure induced by $(X_1,X_2)$ on $\mathbb{R}^2$. Generally $\mathbf{E}[\epsilon_{\hat\eta_1,\hat\eta_2}]$ can never be 0 for a nonzero $\rho$. However, if we compute $\mathbf{E}[{\epsilon_{\eta_1^\star,\eta_2^\star}}]$, by law of total expectation, we have 
\begin{multline}
\mathbf{E}[\epsilon_{\eta_1^\star,\eta_2^\star}] = \frac{1}{2}\Big(\mathbf{E}[X_2]-\mathbf{E}\big[\mathbf{E} [X_2 \mid X_1 = x_1, U=1 ]\big]\Big) \\ + \frac{1}{2}\Big(\mathbf{E}[X_1]-\mathbf{E}\big[\mathbf{E} [X_1 \mid X_2 = x_2, U=2 ]\big]\Big) = 0,
\end{multline} 
which implies that even for Gaussian distributions, the optimal estimator pair $(\eta_1^\star,\eta_2^\star)$ corresponding to the max-scheduling policy can be nonlinear. The max-scheduling policy defines the random variable $U$ in the following sense:
\begin{equation}\label{maxscheduling}
    U = \begin{cases}
        1, &  \big(x_1-\eta_1(x_2)\big)^2 \geq \big(x_2-\eta_2(x_1)\big)^2 \\
        2, &  \big(x_1-\eta_1(x_2)\big)^2 < \big(x_2-\eta_2(x_1)\big)^2.
    \end{cases}
\end{equation}

For any estimator pair $(\eta_1,\eta_2)$, the scheduler always chooses to send out the variable that would give the larger estimation error if not transmitted. If we plug \cref{maxscheduling} into \cref{EV}, it is immediate to see that minimizing $\mathbf{E}[\epsilon_{\eta_1,\eta_2}]$ is equivalent to minimizing the expectation of the error which has a smaller magnitude among those two. Thus, we define the optimality of estimators corresponding to the max-scheduling policy in the following proposition.
\begin{proposition}
    Suppose the scheduler is desinated to apply the \textit{max-scheduling policy} as in \cref{maxscheduling}. Then the corresponding optimal estimators $\eta_1^\star$ and $\eta_2^\star$ are:
    \begin{multline}\label{cost}
        (\eta_1^\star,\eta_2^\star) \in \underset{(\eta_1,\eta_2)\in \mathrm{H}_1 \times \mathrm{H}_2}{\arg\min}  \mathcal{J}(\eta_1,\eta_2) \\ \Equaldef \mathbf{E} \left[ \min \bigg\{ \big(X_1-\eta_1({X}_2)\big)^2 , \big(X_2-\eta_2({X}_1)\big)^2\bigg\}  \right].
    \end{multline}
\end{proposition}
In practice, given a set of independent and identically distributed data samples $\{(x_1^i,x_2^i)\}_{i=1}^{M}$ drawn from $f_{X_1X_2}$, we evaluate the performance of $(\eta_1,\eta_2)$ using the empirical mean:
\begin{equation}\label{EM}
    \bar{\mathcal{J}}(\eta_1,\eta_2) = \frac{1}{M} \sum_{i=1}^{M} \min \Big\{\big(x_1^i-\eta_1(x_2^i)\big)^2 , \big(x_2^i-\eta_2(x_1^i)\big)^2\Big\}.
\end{equation}

\section{Approximation of Nonlinear Estimators}

In general, Problem \eqref{cost} does not admit a solution in closed form \cite{Vasconcelos:2021}. There are no systematic ways of solving this problem in its current description. Instead, we need to use some form of approximation to obtain a candidate pair of optimal solutions $(\eta_1^\star,\eta_2^\star)$. Before we proceed with the search for optimal estimators, one natural question arises: \textit{-- What are the appropriate spaces $\mathrm{H}_1$ and $\mathrm{H}_2$ of admissible estimators?} 

We assume that the admissible estimators belong to the  sets of continuously differentiable functions. So ideally $\mathrm{H}_1 = \mathcal{C}^1[a,b]$ and $\mathrm{H}_2 = \mathcal{C}^1[c,d]$, where 
\begin{align}
& a \Equaldef \min \{x_2^1,\dots,x_2^{M} \}, \ b \Equaldef \max \{x_2^1,\dots,x_2^{M}\} \\ & c \Equaldef \min \{x_1^1,\dots,x_1^{M}\}, \  d \Equaldef \max \{x_1^1,\dots,x_1^{M}\}.
\end{align}
This means that we will ignore the performance of estimator pairs on $\mathbb{R}^2 \setminus [a.b] \times [c,d]$ since there are no samples on this set. For $\mathrm{H}_1 = \mathcal{C}^1[a,b]$ and $\mathrm{H}_2 = \mathcal{C}^1[c,d]$, Problem \eqref{cost} is still non-convex, and infinite-dimensional. To solve it in practice, we must turn it into a finite dimensional optimization problem first. In this section, we will discuss how to  approximately solve Problem \eqref{cost} by performing optimization over a vector of weight parameters.

\subsection{Principles of Linear Function Approximation}




We are going to use a linear combination of functions $\{\phi_n\}_{n=1}^{\infty}$ to approximate the optimal estimator pair $(\eta_1^\star,\eta_2^\star)$. It is known that, rigorously, the cardinality of the set of basis functions for $\mathcal{C}^1[a,b]$ is $\aleph_1$. Therefore, we can not define a countable basis for all differentiable functions. To make the approximation meaningful, we should assume that a unique linear combination of $\{\phi_n\}_{n=1}^{\infty}$ approaches $\eta_1^\star$ and $\eta_2^\star$ in a Schauder basis sense \cite{Albiac:2006}.
That is, there exist two unique sequences of weight parameters $\{w_1^n\}_{n=1}^{\infty}$ and $\{w_2^n\}_{n=1}^{\infty}$ such that 
\begin{equation}
  \eta_1^\star(x) =  \lim\limits_{K \to \infty}\sum_{n=1}^{K+1} w_1^n \phi_n(x),  \ \ x \in [a,b] \end{equation} 
and
  \begin{equation}
  \eta_2^\star(x) = \lim\limits_{K \to \infty} \sum_{n=1}^{K+1} w_2^n \phi_n(x), \ \  x\in[c,d]. 
  \end{equation}
  Here, the convergence is with respect to the sup norm equipped on the space of continuously differentiable functions $\mathcal{C}^1$. Formally,
consider a set of differentiable functions denoted by $\Phi$, such that:
\begin{equation}\label{sob}
\Phi(x) = \begin{bmatrix}
\phi_1(x) \\
\vdots \\
\phi_{K+1}(x) \\
\vdots
\end{bmatrix}.
\end{equation}
Let $\phi_1(x) = 1,$ $x\in\mathbb{R}$. A linear approximation for our estimator functions is given by:
\begin{equation}
\eta_j(x) = \vw_j^\T\Phi(x),
\end{equation}
where $\vw_j$ is a sequence of real numbers, $j \in \{1,2\}$. Then, the optimization is  performed over the weight sequences $\vw\Equaldef (\vw_1,\vw_2)$, as follows:
\begin{multline}\label{Approx}
    \min {\mathcal{J}}(\vw) = \mathbf{E} \left[ \min \bigg\{\big(X_1-\vw_1^\T\Phi({X}_2)\big)^2 , \right. \\  \left. \big(X_2-\vw_2^\T\Phi({X}_1)\big)^2\bigg\}  \right].
\end{multline}
For computational tractability, we must truncate sequences $\vw_1$ and $\vw_2$ to have $(K+1)$ components, where $K$ can be interpreted as a measure of complexity for the approximation architecture. This reduces the dimension of the original problem to $(2K+2)$.
With samples $\{(x_1^i,x_2^i)\}_{i=1}^{M}$, Problem \eqref{Approx} is approximated by:
\begin{multline}\label{ApproxW}
    \min_{\hat\vw \in\mathbb{R}^{2K+2}} \hat{\mathcal{J}}(\hat\vw) = \frac{1}{M} \sum_{i=1}^{M} \min \Big\{\big(x_1^i-\hat{\vw}_1^\T\Phi({x}_2^i)\big)^2 , \\ \big(x_2^i-\hat{\vw}_2^\T\Phi({x}_1^i)\big)^2\Big\}.  
\end{multline}

The choice of $\phi_n$ functions and the dimension must be on a case-by-case analysis. A smart choice of $\{\phi_n\}_{n=1}^{\infty}$ would allow for a shorter truncation, which significantly reduce the computational cost. 
By solving \cref{ApproxW} instead of solving \cref{cost}, we  lose accuracy with respect to two aspects. First, after the truncation, the search is restricted to the space $\mathrm{span}\{\phi_1,\dots,\phi_{2K+1}\}$, which is finite-dimensional. Second, the resulting estimators from \eqref{ApproxW} can be tailored to the specific data set $\{(x_1^i,x_2^i)\}_{i=1}^{M}$. We consider this as a over-fitting: the resulting estimator may perform extremely well on this specific data set, but it is far from the optimum when considering the whole distribution. The former aspect of losing accuracy is inevitable since we do not have infinite computational resources. The later aspect, however, has some patterns. Through our numerical experiments, we observed that, to avoid over-fitting, it is necessary to require resulting estimators to preserve the distribution type. That is, if $X_1$ and $X_2$ are continuous random variables, then $\eta_2(X_1)$ and $\eta_1(X_2)$ should have continuous distributions. We have the following assumption. 
\begin{assumption}\label{ASUMP}
    Let $(X_1,X_2)$ be a continuous random vector supported on $\mathbb{R}^2$ with $\mathrm{Cor}\{X_1,X_2\} \neq 0$. The conditional expectations, $\mathbf{E} \big[X_1 \mid X_2, U=2 \big]$ and $\mathbf{E} \big[X_2 \mid X_1, U=1 \big]$, when regarded as random variables $\eta_1^\star(X_2)$ and $\eta_2^\star(X_1)$, have continuous distributions supported on $\mathbb{R}$. That is, the probability measure on $\mathbb{R}$ induced by $\eta_1^\star(X_2)$ and $\eta_2^\star(X_1)$ should be absolutely continuous with respect to the Lebesgue measure \cite{Cinlar:2011}.
\end{assumption}
The reason for \cref{ASUMP} is very natural: as $\eta_1^\star(X_2)$ and $\eta_2^\star(X_1)$ are estimations for $X_1$ and $X_2$, we should require this comparison happens in the same family of distributions. Unless $X_1$ and $X_2$ are independent, then $X_1$ carries some information about $X_2$ and vice versa. The information can be incomplete, but when it comes to the conditional expectation and change the distribution correspondingly, we assume that the resulting distribution remains continuous.\\





The following theorem provides a basic qualification for functions $\{\phi_n\}_{n=1}^{\infty}$ based on the previous observation.
\begin{theorem}\label{bq}
    Let $(X_1,X_2)$ be a random vector supported on $\mathbb{R}^2$ with continuous distribution. Denote Lebesgue measure on $\mathbb{R}$ by $\lambda(\cdot)$. The following two statements cannot be simultaneously true: \par
       $\text{1.} \hspace{20pt} \lambda \Big(\bigcap_{n=1}^{\infty}\{x \in \mathbb{R} \mid \phi'_n (x) = 0 \}\Big) > 0 $\par
       \vspace{5pt}
       \text{2.} \hspace{15pt} There exist unique sequences $\{w_1^n\}_{n=1}^{\infty}$ and $\{w_2^n\}_{n=1}^{\infty}$ such that \par
       \vspace{5pt}
          \hspace{25pt} $\eta_1^\star(x) =  \lim\limits_{K \to \infty}\sum_{n=1}^{K+1} w_1^n \phi_n(x)  \hspace{10pt} \text{and} \hspace{10pt} \eta_2^\star(x) = \lim\limits_{K \to \infty} \sum_{n=1}^{K+1} w_2^n \phi_n(x).$
       
\end{theorem}
\begin{proof}
    We take random variables $X_1$ and $\eta_1(X_2)$ as an example. The other pair would be similar thus it is omitted here. The optimal estimator $\eta_1^*(\cdot)$ is the conditional expectation $\mathbf{E} \big[X_1 \mid X_2 ,\mathbf{U}=2 \big]$, it induces a probability measure $\mu'$ for sets of $x_2$: $\mu' \Equaldef \mu_{X_1}(\eta_1^{*-1}(\cdot))$, here $\mu_{X_1}$ denotes the probability measure induced by $X_1$ and $\eta_1^{*-1}(\cdot)$ is the pre-image of a set through $\eta_1^*$. Note that by our assumption \ref{ASUMP},  $\mu'$ is a measure induced by a continuous random variable supported on $\mathbb{R}$, by Radon-Nikodym theorem, $(\exists) \hspace{5pt} f,g:\mathbb{R} \rightarrow \mathbb{R}$ positive almost everywhere, such that $(\forall)$ measurable set $A$,
\begin{equation}
       \mu_{X_1}(A) = \int_{A} f \mathrm{d}\lambda \hspace{5pt}\text{and} \hspace{5pt}      \mu'(A) = \int_{A} g \mathrm{d}\lambda
\end{equation}
such $f$ and $g$ are probability density functions. As they are positive almost everywhere on $\mathbb{R}$, functions $\frac{1}{f}$ and $\frac{1}{g}$ are positive and exist almost everywhere on $\mathbb{R}$. We can write:
\begin{equation}
       \lambda(A) = \int_{A} \frac{1}{f} \mathrm{d}\mu_{X_1} \hspace{5pt}\text{and} \hspace{5pt}      \lambda(A) = \int_{A} \frac{1}{g} \mathrm{d}\mu'
\end{equation}
This equation indicates equivalence between measures $\mu_{X_1}$, $\mu'$ and $\lambda$.\\
Now consider $\bar\eta_1 = \sum_{n=1}^{\infty} w_1^n \phi_n$ as the optimizer of problem \ref{ApproxW} when $K \rightarrow \infty$. Similarly we define $\mu'' \Equaldef \mu_{X_1}(\bar\eta_1^{-1}(\cdot))$. $\mu''$ is the probability measure on $\mathbb{R}$ induced by random variable $\bar\eta_1(X_2)$. On the set $\cap_{n=1}^{\infty}\{x \in \mathbb{R} \mid \phi'_n (x) = 0 \}$, $\phi_n$ can only take a constant value, and the weight parameter sequence $\{w_1^n\}_{n=1}^{\infty}$is unique. This implies that $\bar\eta_1 = \bar c$ on $\cap_{n=1}^{\infty}\{x \in \mathbb{R} \mid \phi'_n (x) = 0 \}$, where $\bar c$ is a constant or $\pm\infty$.\\
On the one hand, $\mu''(\{\bar c\}) \geq \mu_{X_1}(\cap_{n=1}^{\infty}\{x \in \mathbb{R} \mid \phi'_n (x) = 0 \})$ by its definition, and we can compute the right hand side by
\begin{equation}\label{IntM}
    \mu_{X_1}(\cap_{n=1}^{\infty}\{x \in \mathbb{R} \mid \phi'_n (x) = 0 \}) = \int_{\cap_{n=1}^{\infty}\{x \in \mathbb{R} \mid \phi'_n (x) = 0 \}} f \mathrm{d}\lambda 
\end{equation}
Suppose $ \lambda (\cap_{n=1}^{\infty}\{x \in \mathbb{R} \mid \phi'_n (x) = 0 \}) > 0$. The integration in \eqref{IntM} is strictly positive since we are integrating a positive a.e. function with respect to Lebesgue measure on a set of positive measure. This shows that $\mu''(\{\bar c\}) > 0$ if $ \lambda (\cap_{n=1}^{\infty}\{x \in \mathbb{R} \mid \phi'_n (x) = 0 \}) > 0$.\\
On the other hand, we have $\lambda(\{\bar c\}) = 0$. Measure $\mu''$ disagrees with $\lambda$ at least on one $\lambda$-null set, thus they can not be equivalent. But $\mu'$ is equivalent with $\lambda$. By looking at the definition of $\mu'$ and $\mu''$, the only explanation is: $\eta_1^*(\cdot)$ and $\bar\eta_1(\cdot)$ are not the same function. The approximation fails in this case even when $K \rightarrow \infty$.
\end{proof}

\begin{corollary}
    Theorem \ref{bq} suggests that simple, wavelets and Rectified Linear Unit (ReLU) functions are not good choices for $\{\phi_n\}_{n=1}^{\infty}$.
\end{corollary}

\begin{remark}
    Theorem \ref{bq} is a necessary condition for the choice of $\phi_n$, whose purpose is to avoid over-fitting. If we  approximate the estimator without taking this condition into consideration, we would obtain an oscillating function that tries to fit every sample point in the data set. 
    Such estimators fail to pass  validation tests: If we generate another set of data from the same distribution and apply our previous estimator, the performance is degraded with high probability.
    Neural networks are popular architectures for approximation of unknown functions and one of the most  commonly used activation functions is the ReLU. Unfortunately, using an architecture based on ReLUs  
   violates the condition proposed in Theorem \ref{bq}. This means the approximation generated by a ReLU network would lead to over-fitting. Thus, we must resort to other activation functions when applying a Neural Network on Problem \eqref{ApproxW}.
\end{remark}

\section{Numerical Results}

As pointed out in the previous section, we should avoid to use a set of functions such that their derivatives being zero  on a set of positive Lebesgue measure. We tried different types of functions, it turns out that the family of softplus functions and polynomials return extrodinary and stable results.  
\subsection{Approximation Using Soft-Plus Functions }
With a slightly abuse of terminology, we call the function \begin{equation}
\phi_{\alpha,\beta}(x) \Equaldef \frac{1}{\alpha} \log\Big(1+\exp\big(1 \pm \alpha (x-\beta)\big)\Big)
\end{equation}
as a soft-plus function since it is a smooth approximation to the ReLU function without violating the principle proposed in Theorem \ref{bq}. Although the linear combination of $\phi_{\alpha,\beta}(x)$ can not cover $\mathcal{C}^1$ for arbitrary countable collection of $(\alpha,\beta)$, it gives us stable and promising results, especially when $(X_1,X_2)$ is a bivariate Gaussian vector. As shown in \cite{4403040}, the density function of the minimum of two correlated Gaussian random variables has an explicit analytical form, which is related to the Gaussian error function $\mathrm{erf}(\cdot)$. 
\\

\begin{figure}[ht]
    \centering
\includegraphics[width=\columnwidth]{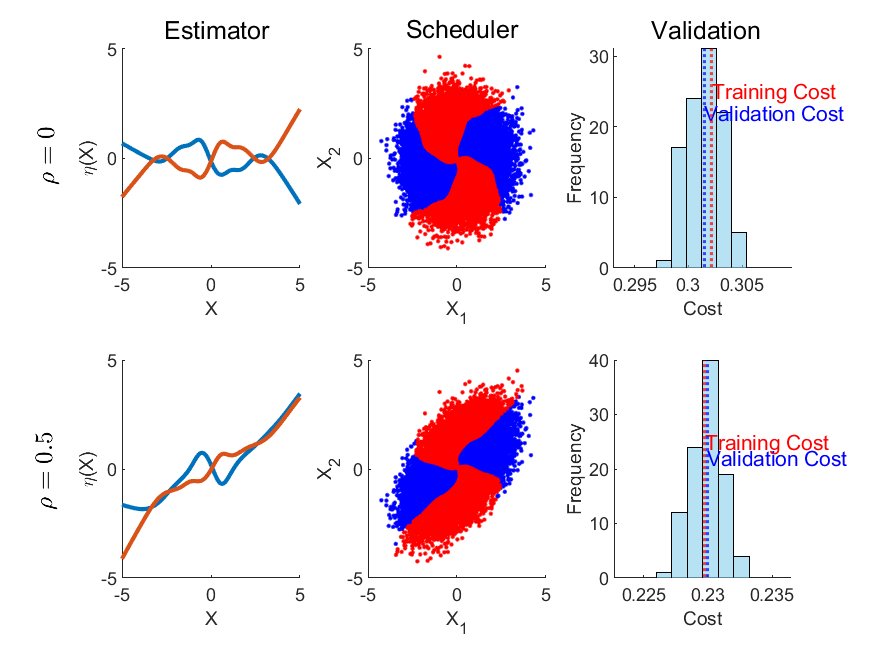}
    \caption{SoftPlus estimator, along with its corresponding scheduling policy, and a validation analysis for a symmetric Gaussian distribution across different correlation values.}
    \label{fig:softplus_results}
\end{figure}

\subsection{Approximation Using Polynomials}

The Weierstrass Approximation Theorem states that on a closed interval, we can approximate a continuous function with vanishing error by increasing the dimension of approximation. The convergence is even uniform for every element on the designated interval. This gives us a theoretical guarantee of convergence if we choose $\phi_n = x^n$. It is also worth mentioning that the polynomials $\{x^n\}_{n=1}^{\infty}$ never violate the condition  $\lambda (\cap_{n=1}^{\infty}\{x \in \mathbb{R} \mid \phi'_n (x) = 0 \}) = 0$. This makes polynomials a great choice. However, in practice, we find out that it requires a relatively large amount of polynomials to get a decent performance on the resulting estimator pair. The idea of polynomial approximation definitely worth more investigation, since it may leads to convergence speed and sample complexity results in the future.\\

\begin{figure}[ht]
    \centering
\includegraphics[width=\columnwidth]{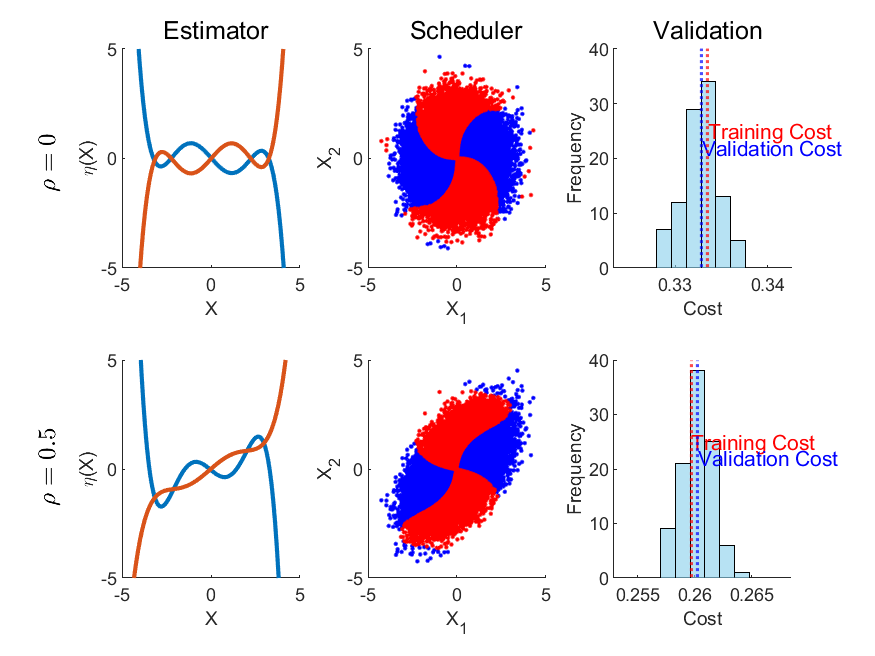}
    \caption{Polynomial estimator, along with its corresponding scheduling policy, and a validation analysis for a symmetric Gaussian distribution across different correlation values.}
    \label{fig:poly_results}
\end{figure}

\begin{table}[ht]
\centering
\caption{Validation cost of four different estimators for the symmetric Gaussian with different correlation values.}
\label{tab:simulations_results}
\begin{tabular}{ccccc}
\hline \hline
$\rho$ & MMSE & Linear & SoftPlus & Polynomial \\ \hline \hline
$0$    & 0.363 & 0.363 & 0.303 & 0.334 \\
$0.25$ & 0.360 & 0.345 & 0.285 & 0.315 \\
$0.5$  & 0.337 & 0.288 & 0.230 & 0.261 \\
$0.75$ & 0.254 & 0.173 & 0.140 & 0.159 \\ \hline \hline
\end{tabular}
\end{table}

\section{Conclusion and Future Work}



We have considered the problem of remote estimation of a bivariate Gaussian source
when only one of the variables is revealed by a scheduler to the estimator. The optimal scheduler for this seemingly simple problem is unknown. Therefore, we propose to look for good solution using a single-layer neural network, called neuro-schedulers. We establish a  condition to determining which activation functions are cannot be used if we are trying to avoid data over-fitting. We also show that the the optimal scheduler is a nonlinear in general. Future work consists of using deep neural-networks and obtaining performance guarantees for specific architectures as well as optimization algorithms that exploit the inherent structure of the scheduling problem.











\bibliography{./reference/ref}
\bibliographystyle{ieeetr}

\end{document}